\documentclass{llncs}

\usepackage{makeidx}  

\usepackage{graphicx}
\usepackage{amssymb}
\usepackage{array}
\usepackage{subfig}
\usepackage{color}
\usepackage{soul} 

\newcommand{\omt}[1]{}

\begin{document}


\title{Improving Resource Location with Locally Precomputed Partial Random Walks\thanks{This research was supported in part by Comunidad de Madrid grant S2009TIC-1692, Spanish MICINN grant TEC2011-29688-C02-01, Spanish MEC grant TIN2011-28347-C02-01 and Bancaixa grant P11B2010-28.}}


\author{V\'{\i}ctor M. L\'opez Mill\'{a}n\inst{1} \and
        Vicent Cholvi\inst{2} \and 
        Luis L\'opez\inst{3} \and\\
        Antonio Fern\'andez Anta\inst{4}}
        
\institute{Universidad CEU San Pablo, Spain 
\email{vmlopez.eps@ceu.es}
\and
Universitat Jaume~I, Spain
\email{vcholvi@uji.es}
\and
Universidad Rey Juan Carlos, Spain
\email{llopez@gsyc.es}
\and
Institute IMDEA Networks, Spain
\email{antonio.fernandez@imdea.org}
}

\maketitle

\begin{abstract}
Random walks can be used to search complex networks for a desired resource. To reduce search lengths, we propose a mechanism based on building random walks connecting together partial walks (PW) previously computed at each network node. Resources found in each PW are registered. Searches can then jump over PWs where the resource is not located. However, we assume that perfect recording of resources may be costly, and hence, probabilistic structures like Bloom filters are used. Then, unnecessary hops may come from false positives at the Bloom filters. 
Two variations of this mechanism have been considered, depending on whether we first choose a PW in the current node and then check it for the resource, or we first check all PWs and then choose one. In addition, PWs can be either simple random walks or self-avoiding random walks. Analytical models are provided to predict expected search lengths and other magnitudes of the resulting four mechanisms. Simulation experiments validate these predictions and allow us to compare these techniques with simple random walk searches, finding very large reductions of expected search lengths.

\end{abstract}



\newcommand{\ER}	{Erd\H{o}s-R\'{e}nyi}

\newcommand{\s}		{\ensuremath{s}}		
\newcommand{\sopt}	{\ensuremath{s_{opt}}}		
\newcommand{\p}		{\ensuremath{p}}		
\newcommand{\W}		{\ensuremath{W}}		
\newcommand{\w}		{\ensuremath{w}}		
\newcommand{\lsave}	{\ensuremath{\overline{l}_s}}	
\newcommand{\lave}	{\ensuremath{\overline{l}}}	
\newcommand{\lopt}	{\ensuremath{\overline{l}_{opt}}} 

\newcommand{\Prv}	{\ensuremath{P}}		
\newcommand{\Pexp}	{\ensuremath{\overline{P}}}	
\newcommand{\Jrv}	{\ensuremath{J}}		
\newcommand{\Jexp}	{\ensuremath{\overline{J}}}	
\newcommand{\Trv}	{\ensuremath{T}}		
\newcommand{\Texp}	{\ensuremath{\overline{T}}}	
\newcommand{\Urv}	{\ensuremath{U}}		
\newcommand{\Uexp}	{\ensuremath{\overline{U}}}	
\newcommand{\Lsrv}	{\ensuremath{L_s}}	
\newcommand{\Lsexp}	{\ensuremath{\overline{L}_s}}		
\newcommand{\Lrv}	{\ensuremath{L}}		
\newcommand{\Lexp}	{\ensuremath{\overline{L}}}		
\newcommand{\Lexpopt}	{\ensuremath{\overline{L}_{opt}}}		

\newcommand{\N}		{\ensuremath{N}}		
\newcommand{\kave}	{\ensuremath{\overline{k}}}	

\newcommand{\Exp}[1]	{\ensuremath{\mathrm{E}[#1]}}
\newcommand{\Prob}[1]	{\ensuremath{\mathrm{P_r}[#1]}}

\newcommand{\Cpw}	{\ensuremath{C_p}}	
\newcommand{\Cs}	{\ensuremath{C_s}}	
\newcommand{\nsearches} {\ensuremath{b}}

\newcommand{\pn} 	{\ensuremath{p_n}}
\newcommand{\ps} 	{\ensuremath{p_{tp}}}
\newcommand{\pf} 	{\ensuremath{p_{fp}}}
\newcommand{\pres} 	{\ensuremath{p_{r}}}

\newcommand{\nk}	{\ensuremath{n_k}}

\section{Introduction}

A \emph{random walk} in a network is a routing mechanism that chooses the next node to visit 
at random among the neighbors of the current node. Random walks have been extensively studied in mathematics,
and have been used in a wide range of applications such as statistic physics, population dynamics, bioinformatics, etc.
When applied to communication networks, random walks 
have had a profound impact on algorithms and complexity theory.
Some of the advantages of random walks  are their simplicity, their small processing power consumption at the nodes, and the fact that they need only local information, avoiding the communication overhead necessary in other routing mechanisms.
An important application of random walks has been the search for resources held in the nodes of a network, also known as the \emph{resource location problem}.
Roughly speaking, the problem consists of finding a node that holds the resource, starting at some \emph{source node}. Random walks can be used to perform such a search as follows. It is checked first if the source node holds the resource. If it does not, the search hops to a random neighbor, that repeats the process. The search proceeds through the network in this way until a node that holds the resource is found.
Due to the random nature of the walk, some nodes may be visited more than once (unnecessarily from the search standpoint), while other nodes may remain unvisited for a long time. The number of hops taken to find the resource is called the \emph{search length} of that walk. The performance of this direct application of random walks to network search has been studied in~\cite{rw:Adamic01,rw:Lv02b,rw:Yang05,rw:Gkantsidis06,rw:Rodero10}.

The use of random walks for resource location has several clear applications, like unstructured peer-to-peer (P2P) file sharing systems or content-centric networks (CCN) \cite{DBLP:journals/cacm/JacobsonSTPBB12}. 
The latter are networks in which the key elements are named content chunks, which are requested by users using the content name. Content chunks have to be efficiently located and transferred to be consumed by the user. The techniques described in this paper could be used in the context of CCN to locate content chunks.

%
%

\paragraph{Contributions}

This paper proposes an application to resource location of the technique of concatenating partial walks (PW) available at each node to build random walks. A PW is a precomputed random walk of fixed length. Two variations are considered, depending on whether the search mechanism first randomly chooses one of the PWs in the current node and then checks its associated information for the desired resource, or it first checks all PWs in the node and then randomly chooses among those with a positive result. Both of these variations may use PWs that are simple random walks (RW) or self-avoiding random-walks (SAW), resulting in four mechanisms referred to as \emph{choose-first} PW-RW or PW-SAW, and \emph{check-first} PW-RW or PW-SAW, respectively.
Our mechanisms assume the use of Bloom filters~\cite{nets:Broder04} to efficiently store the set of resources (not their owners) held by the nodes in each partial walk. The compactness of Bloom filters comes at the price of possible \emph{false positives} when checking if a given resource is in the partial walk. False positives occur with a probability $\p$, which is taken into account in our analyses.
These assumptions provide generality to our model, since a probability of $\p=0$ models the case in which the full list of resources found are stored (instead of using a Bloom filter).

We provide an analytical model for the choose-first PW-RW technique, with
expressions for the \emph{expected search length}, the \emph{optimal length of the partial walks}, and for the \emph{optimal expected search length}. We found that, when the probability of false positives in Bloom filters is small, the optimal expected search length is proportional to the square root of the expected search length achieved by simple random walks, in agreement with the results in~\cite{rw:DasSarma13}. Another interesting finding is that the optimal length of the partial walks does not depend on the probability of false positives of the Bloom filters. 
We also provide analytical models for the choose-first PW-SAW mechanism as well as for the check-first variations, which predict their expected search length. 
Then, the predictions of the models are validated by simulation experiments in three types of randomly built networks: regular, \ER, and scale-free. These experiments are also used to compare the performance of the four mechanisms, and to investigate the influence of parameters as the false positive probability and the number of partial walks per node. 
Finally, we have compared the performance of the four search mechanisms with respect to simple random walk searches. For choose-first PW-RW we have found a reduction in the average search length ranging from around $98\%$ to $88\%$. For choose-first PW-SAW such a reduction is even bigger, ranging from $12\%$ to $5\%$ with respect to PW-RW. Check-first PW-RW and PW-SAW can achieve still larger reductions increasing the number of PWs available at each node.

\paragraph{Related Work.}


Das Sarma et al.~\cite{rw:DasSarma13} proposed a distributed algorithm to obtain a random walk of a specified length $\ell$ in a number of rounds\footnote{A \emph{round} is a unit of discrete time in which every node is allowed to send a message to one of its neighbors. According to this definition, a simple random walk of length $\ell$ would then take $\ell$ rounds to be computed.} proportional to $\sqrt{\ell}$. In the first phase, every node in the network prepares a number of \emph{short (random) walks} departing from itself. The second phase takes place when a random walk of a given length starting from a given source node is requested. One of the short walks of the source node is randomly chosen to be the first part of the requested random walk. Then, the last node of that short walk is processed. One of its short walks is randomly chosen, and it is \emph{connected} to the previous short walk. The process continues until the desired length is reached.


Hieungmany and Shioda~\cite{rw:Hieungmany10} proposed a \emph{random-walk-based} file search for P2P networks. A search is conducted along the concatenation of hop-limited shortest path trees.  To find a file, a node first checks its \emph{file list} (i.e., an index of files owned by neighbor nodes). If the requested file is found in the list, the node sends the file request message to the file owner. Otherwise, it randomly selects a leaf node of the hop-limited shortest path tree, and the search follows that path, checking the \emph{file list} of each node in it.


The use of partial random walks in resource location has been proposed in~\cite{rw:LopezMillan12_TADDS} for networks with dynamic resoures. Our work in this paper incorporates efficient storage by means of Bloom filters, in the context of static resources. The use of SAWs as PWs is also proposed and compared with simple RWs.

\paragraph{Structure.} The next section presents a model for the four search mechanisms proposed. Then, the choose-first PW-RW is evaluated in Section~\ref{sec:choose-first_PWRW}. For the sake of clarity, the choose-first PW-SAW mechanism is covered separately in Section~\ref{sec:lave_SAW}, which includes the corresponding analysis together with performance results. Similarly, the check-first PW-RW/PW-SAW mechanisms are presented in Section~\ref{sec:check-first}. 

\section{Model}
\label{sec:analytical}



Let us consider a randomly built network of $N$ nodes
and arbitrary topology, whose nodes hold resources randomly placed in them. Resources are unique, i.e., there is a single instance of each resource in the network. The resource location problem is defined as visiting the node that holds the resource, starting from a certain node (the \emph{source} node). For each search, the source node is chosen uniformly at random among all nodes in the network. 

The search mechanisms proposed in this paper exploit the idea of efficiently building \emph{total random walks} from \emph{partial random walks} available at each node of the network. This process comprises two stages:

\paragraph{(1) Partial walks construction.} Every node $i$ in the network precomputes a set $\W_i$ of \w\ random walks in an initial stage before the searches take place. Each of these partial walks has length \s, starting at $i$ and finishing at a node reached after \s\ hops. In the PW-RW mechanism, the partial walks computed in this stage are simple random walks.
During the computation of each partial walk in $\W_i$, node $i$ registers the resources held by the \s\ first nodes in the partial walk (from $i$ to the one before the last node). As mentioned, for generality, we assume that the resources found are stored in a Bloom filter. 
This information will be used in Stage 2.
Bloom filters are space-efficient randomized data structures to store sets, supporting membership queries. Thus, the Bloom filter of a partial walk can be queried for a given resource. If the result is negative, the resource is not in any of the nodes of the partial walk. If the result is positive, the resource is in one of the nodes of the partial walk, unless the result was a \emph{false positive}, which occurs with a certain probability \p.\footnote{More concretely, \p\ is the probability of obtaining a positive result conditioned on the desired resource not being in the filter.} The size of the Bloom filters can be designed for a target (small) \p\ considered appropriate.
A variation of the partial walk construction mechanism consists of using PWs that are \emph{self-avoiding} walks (SAW). The resulting mechanism, called PW-SAW, is 
analyzed in Section~\ref{sec:lave_SAW}.

\paragraph{(2) The searches.} After the PWs are constructed, searches are performed in the following fashion when the choose-first PW-RW/PW-SAW mechanisms are used. When a search starts at a node $A$, a PW in $\W_A$ is chosen uniformly at random. Its Bloom filter is then queried for the desired resource. If the result is negative, the search \emph{jumps} to node $B$, the last node of that partial walk. 
The process is then repeated at $B$, so that the search 
keeps jumping in this way while the results of the queries are negative. When at a node $C$, the query to the Bloom filter (of the PW randomly chosen from $\W_C$) gives a positive result, the search \emph{traverses} that partial walk looking for the resource
until the resource is found or the partial walk is finished. If the resource is found, the search stops. If the search reaches the last node $D$ of the partial walk without having found the resource in the previous nodes, it means that the result of the Bloom filter query was a false positive. The search then randomly chooses a partial walk in $\W_D$ and decides whether to jump over it or to traverse it depending on the result of the query to its Bloom filter, as described above. 
A variation of this behavior consists of first checking all PWs of the node for the desired resource, and then randomly choosing among the ones with a positive result. The resulting mechanisms, called check-first PW-RW/PW-SAW are analyzed in Section~\ref{sec:check-first}.
\\[0.1cm]



In this work, we are interested in the number of \emph{hops} to find a resource (when PWs of length \s\ are used), which is defined as the \emph{search length} and denoted \Lsrv.
Some of these hops are \emph{jumps} (over PWs) and other are \emph{steps} (traversing PWs). In turn, we distinguish between \emph{trailing steps}, if they are the ones taken when the resource is found, and \emph{unnecessary steps}, if they are taken when the resource is not found.
The search length is a random variable that takes different values when independent searches are performed. The \emph{search length distribution} is defined as the probability distribution of the search length random variable. We are interested in finding the \emph{expected search length}, denoted \Lsexp. Figure~\ref{fig:concepts} summarizes the behavior of the search mechanisms.

At this point, we emphasize the difference between the \emph{search} just defined and the \emph{total walk} that supports it, consisting of the concatenation of \emph{partial walks} as defined above. Searches are shorter in length than their corresponding total walks because of the number of steps saved in jumps over partial walks in which we know that the resource is not located (although these saving may be reduced by the unnecessary steps due to Bloom filter false positives).

\begin{figure}
 \centering
 \includegraphics[width=10cm]{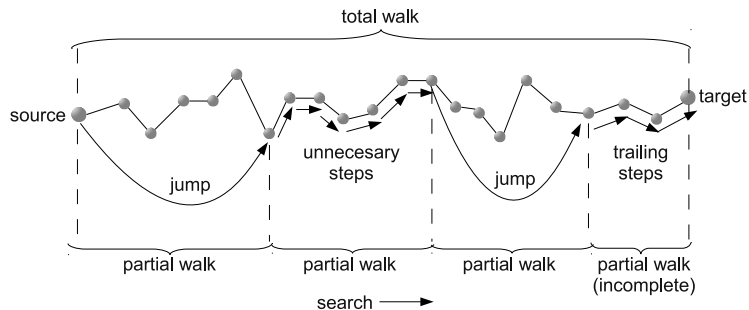}
 \caption{An example of search, using PWs of length $\s=6$.}
 \label{fig:concepts}
\end{figure} 

\section{Choose-First PW-RW}
\label{sec:choose-first_PWRW}

\subsection{Analysis of Choose-First PW-RW}
\label{sec:lave_RW}

We make an additional assumption in order to simplify this analysis. Once a PW has been used in the total walk of a search, it is never reused again in that total walk or in any other searches. 
Thus we guarantee that the total walks are true random walks. 
This implies that in practice each node needs to have a large number of precomputed partial walks (\w), assumption that would compromise the benefits of the proposed mechanism in practice. Simulations in Section~\ref{sec:perf_eval} show that real cases with small \w\ behave very similarly to the base case provided by this analysis.

Let \Lsrv\ be the random variable representing the number of hops in the search (i.e., its length) when PWs of length \s\ are used. The expected search length is denoted by \Lsexp. Let \Lrv\ be the random variable representing the number of hops of the corresponding total walk. Its expected search length is denoted \Lexp. Making use of the assumption that partial walks are never reused, \Lrv\ can be viewed as the length of a search based on a simple random walk in the considered network, and \Lexp\ as the expected search length of random walks in that network. Then, we can state the following theorem: 



\begin{theorem}
\label{t-expectedL}
If the expected number of trailing steps is assumed to be uniformly distributed in $[0,s-1]$\footnote{This is, in fact, a pessimistic assumption. The distribution of trailing steps is approximately uniform, but shorter walks have a slightly higher probability than longer ones. This can be shown analytically and has been confirmed in our experiments (see Appendix~\ref{ssec:distr_tsteps}). Therefore, the expected value in our analysis, derived from a perfectly uniform distribution, is slightly higher than the real average value.}, then the expected search length is:
\begin{equation}
\Lsexp = \left(\frac{\s}{2} + \frac{2\Lexp +1}{2\s} - 1 \right)\cdot (1-\p) + \Lexp\cdot p.
\label{eq:Lsexp_p}
\end{equation}
\end{theorem}

\begin{proof}
Let \Prv, \Jrv, \Urv\ and \Trv\ be random variables representing the number of partial walks, jumps, unnecessary steps and trailing steps in a search, respectively. Their expectations are denoted as \Pexp, \Jexp, \Uexp\ and \Texp.
Since hops in a search can be jumps, unnecessary steps or trailing steps, it follows that,
$\Lsrv = \Jrv + \Urv + \Trv.$
Then, the expected search length for partial walks of size \s\ is\footnote{In the following, we make implicit use of the linearity properties of expectations of random variables.}
$\Lsexp = \Jexp + \Uexp + \Texp.$

The expected number of jumps can be obtained from the expected number of partial walks in the search (\Pexp) and from the probability of false positive (\p) as
$\Jexp = \Pexp\cdot(1-p),$
since \Jrv\ follows a binomial distribution $B(\Prv,1-\p)$, where the number of experiments is the random variable representing the number of partial walks in a search (\Prv) and the success probability is the probability of obtaining a negative result in a Bloom filter query ($1-\p$).\footnote{If $Y$ is a random variable with a binomial distribution with success probability $p$, in which the number of experiments is in turn the random variable $X$, it can be easily shown that $\overline{Y} = \overline{X}\cdot p$ (see Appendix~\ref{sec:demo_exp}).} 

For the expected number of unnecessary steps,
$\Uexp = \Pexp\cdot \p \cdot \s,$
since $\Pexp \cdot \p$ is the expected number of false positives in the search and each of them contributes with \s\ unnecesary steps.
The number of partial walks in a search can be obtained dividing the length of the total walk by the size of a partial walk: $\Prv = \left\lfloor \frac{\Lrv}{\s} \right\rfloor = \frac{\Lrv - \Trv}{\s}$. Then, the expected number of partial walks in a search is
$\Pexp = \frac{\Lexp - \Texp}{\s}.$

Since we assume that the expected number of trailing steps is uniformly distributed between $0$ and ($\s-1$), its expectation is
$\Texp = \frac{\s-1}{2}.$

Using the previous equations 
we have:
\begin{equation}
 \Lsexp = \left(\frac{\s}{2} + \frac{2\Lexp +1}{2\s} - 1 \right) + \p\cdot\left( \Lexp - \left(\frac{\s}{2} + \frac{2\Lexp +1}{2\s} - 1\right) \right),
\end{equation}
where the first term is the expectation of the search length for a ``perfect'' Bloom filter (one that never returns a false positive when the resource is not in the filter, i.e., $\p=0$), and the second term is the expectation of the additional search length due to false positives ($p\neq0$).

Another interpretation of this expression is obtained if we reorganize it to make explicit the contributions of a perfect filter and of a ``broken'' filter (one that always returns a false positive result when the resource is not in the filter, i.e., $\p=1$) as 
\begin{equation}
 \Lsexp = \left(\frac{\s}{2} + \frac{2\Lexp +1}{2\s} - 1 \right)\cdot (1-\p) + \Lexp\cdot p.
\end{equation}

\end{proof}

From this theorem and using calculus, 
we have the following corollary.

\begin{corollary}
The optimal length of the partial walks, i.e., the length of the partial walks that minimizes the expected search length, is:
\begin{equation}
 \sopt = \sqrt{2\Lexp + 1}.
 \label{eq:s_opt}
\end{equation}
\end{corollary}

\noindent The obtained value needs to be rounded to an integer, which is omitted in the notation.
Observe that \emph{the optimal length of the partial walks is independent from the probability of false positives in the Bloom filters}, while the expected search length (\Lsexp) does of course depend on it.

\begin{corollary}
The optimal expected search length, i.e., the expected search length when partial walks of optimal length are used, is:
\begin{equation}
 \Lexpopt = \left(\sqrt{2\Lexp + 1} - 1\right) \, (1 - \p) + \Lexp\,\p= (\sopt - 1) \, (1 - \p) + \Lexp\,\p.
\label{eq:lopt}
\end{equation}
\end{corollary}
This result is an interesting relation between the optimal length of the search and the optimal length of the PWs.
If we consider perfect Bloom filters ($\p = 0$), we have $\Lexpopt = \sopt - 1,$
which for large \Lexp\ (e.g. for large networks) becomes $\Lexpopt \approx \sopt$. Therefore, we have found that, for large $N$ and $\p=0$, \emph{the optimal expected search length approximately equals the optimal length of the partial walks. }
For arbitrary values of \p, Equation~\ref{eq:lopt} shows that \Lexpopt\ is linear in \p. 

This completes the analysis of choose-first PW-RW. Appendix~\ref{sec:alt-analysis-PWRW1} provides an alternative analysis using a different approach. Instead of assuming that the total walk is a random walk, it considers that it is built using the $w$ PWs available at each node, which avoids the need of \Lexp. On the other hand, the alternative model does not provide expressions for the optimal PW length or the expected search length.

\subsection{Cost of Precomputing PWs}

Since searches use the partial walks precomputed by each of the nodes of the network, the cost of this computation must be taken into account. We measure this cost as the number of messages \Cpw\ that need to be sent to compute all the PWs in the network. This quantity has been chosen to be consistent with our measure of the performance of the searches. Indeed, each \emph{hop} taken by a search can be alternatively considered as a \emph{message} sent. In addition, $\Cpw$ is independent from other factors like the processing power of nodes, the bandwidth of links and the load of the network.
The cost of precomputing a set of PWs can be simply obtained as $\Cpw = \N \w (\s+1)$, since each of the \N\ nodes in the network computes \w\ partial walks, sending \s\ messages to build each of them plus one extra message to get back to its source node. 

Let's suppose that each node starts on the average \nsearches\ searches that are processed by the network with the set of PWs precomputed initially. We define \Cs\ to be the total number of messages needed to complete those searches. If the expected number of messages of a search is $\Lsexp+1$ (counting the message to get back to the source node), we have that $\Cs = N \nsearches (\Lsexp+1)$. Now, defining $C_t$ as the \emph{average total cost per search}, we can write:

\begin{equation}
 C_t = \frac{\Cs+\Cpw}{N\nsearches} = (\Lsexp+1) + \frac{w}{b}(\s+1).
 \label{eq:Ct}
\end{equation}

The second term in Equation~\ref{eq:Ct} is the contribution to the cost of the precomputation of the PWs. This contribution will remain small provided that the number of searches per node in the interval is large enough.

\subsection{Performance Evaluation}
\label{sec:perf_eval}

The goal of this section is to apply the model for choose-first PW-RW presented in the previous section to real networks, and to validate its predictions with data obtained from simulations. Three types of networks have been chosen for the experiments: regular networks (constant node degree), \ER\ (ER) networks and scale-free networks (with power law on the node degree). A network of each type and size $N=10^4$ has been randomly built with the method proposed by Newman et al.~\cite{nets:Newman01} for networks with arbitrary degree distribution, setting their average node degree to $\kave=10$. Each network is constructed in three steps: (1) a preliminary network is constructed according to its type; (2) its degree distribution is extracted, and (3) the final (random) network is obtained feeding the Newman method with that degree distribution. For each experiment, $10^6$ searches have been performed, with the source node chosen uniformly at random among the $N$ nodes. Likewise, the resource has been placed in a node chosen uniformly at random for each experiment.

\subsubsection{Optimal PW Size and Expected Search Length in Choose-First PW-RW}

We start by applying 
Theorem~\ref{t-expectedL} to the networks described above to obtain the expected search length as a function of the size of the PWs.\footnote{For each network, the expected length of a random walk search (\Lexp) is needed. We estimate these expected values by simulating $10^6$ simple random walk searches and averaging their lengths in each of the networks (these average search lengths are denoted using lowercase ($\overline{l}$) to distinguish them from the actual expected value (\Lexp) in the model. The values obtained from the experiments are: $\overline{l}_{reg} = 11246$, $\overline{l}_{ER} = 12338$, and $\overline{l}_{sf} = 15166$).
These results agree with the approximate analytical method in~\cite{rw:LopezMillan12_Networks} (a modification of the one provided in~\cite{rw:Rodero10}), which produces the following results: $\overline{l}_{reg} = 11095$, $\overline{l}_{ER} = 12191$, and $\overline{l}_{sf} = 14920$. 
}
Figure~\ref{fig:figA}\subref{fig:ls_s} provides plots of the expected search lengths (\Lsexp) given by Equation~\ref{eq:Lsexp_p} as a function of the size of the PWs (\s), when the probability of a false positive in the Bloom filter is set to $\p=0$, for the three types of networks considered. Results from the analytical model are shown as curves while simulation data are shown as points. The curves for the three networks show a minimum point $(\sopt,\Lexpopt)$. This behavior is due to the fact that, when $\s$ is small, the number of jumps needed to reach a PW containing the chosen resource grows, therefore increasing the value of $\Lexp$. In turn, for larger values of $\s$, the number of trailing steps within the last PW grows, also increasing the value of $\Lexp$.



\begin{figure*}[!t]
\centerline{\subfloat[]{\includegraphics[width=7.8cm]{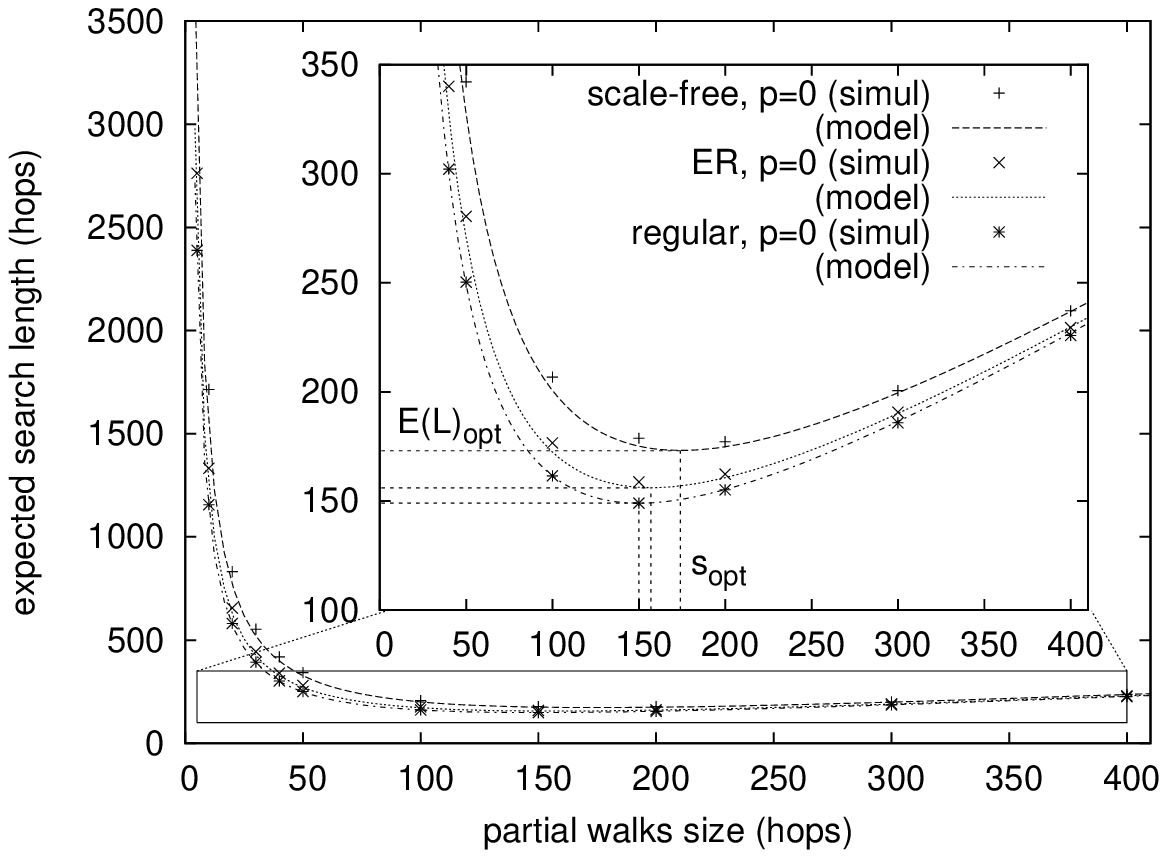} \label{fig:ls_s}} \hfil
\subfloat[]{\includegraphics[width=7.8cm]{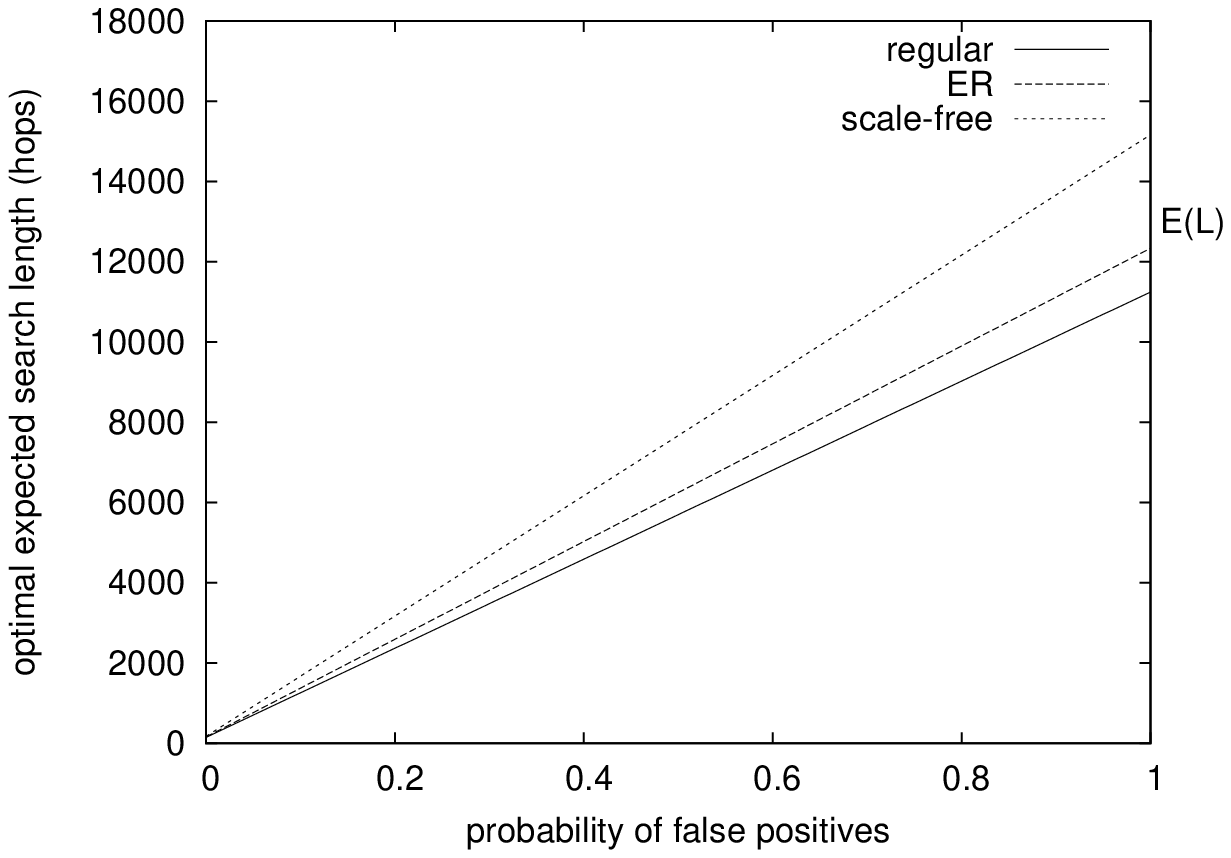} \label{fig:lopt_p}}
}
\caption{(a) Expected search length (\Lsexp) as a function of \s\ when $\p=0$ in a regular network, an ER network and a scale-free network. The optimal points $(\sopt,\Lexpopt)$ for each network are $(150,149)$, $(157,156)$, and $(174,173)$. (b) Optimal expected search length (\Lexpopt) as a function of \p.}
\label{fig:figA}
\end{figure*}

%


Figure~\ref{fig:figA}\subref{fig:lopt_p} illustrates (using Equation~\ref{eq:lopt} and taking into account the fact that \sopt\ is independent from the value of \p) the optimal expected search length (\Lexpopt) as a function of the probability of false positives (\p). It can be seen that it grows linearly: the regular network exhibits the smallest slope, followed by the ER network and then by the scale-free network. For $\p=1$, Equation~\ref{eq:lopt} degenerates to $\Lexpopt = \Lexp$, since the search performs all the hops of the total walk (i.e., it is a random walk). In fact, Equation~\ref{eq:Lsexp_p} also degenerates to $\Lsexp = \Lexp$ in this case, meaning that the expected search length is that of random walk searches regardless the size of the PWs (\s).


\subsubsection{Distributions of Search Lengths in Choose-First PW-RW}
\label{sec:search_length_distr}

The aim of this section is to experimentally explore how the use of PWs affects the statistical distribution of search lengths.

\begin{figure*}[!t]

 
\centerline{
 \subfloat[Search lengths for $\p=0$ and for $\s=\sopt=150$, $\s=50$ and $\s=1000$.]{
 \includegraphics[width=7.8cm]{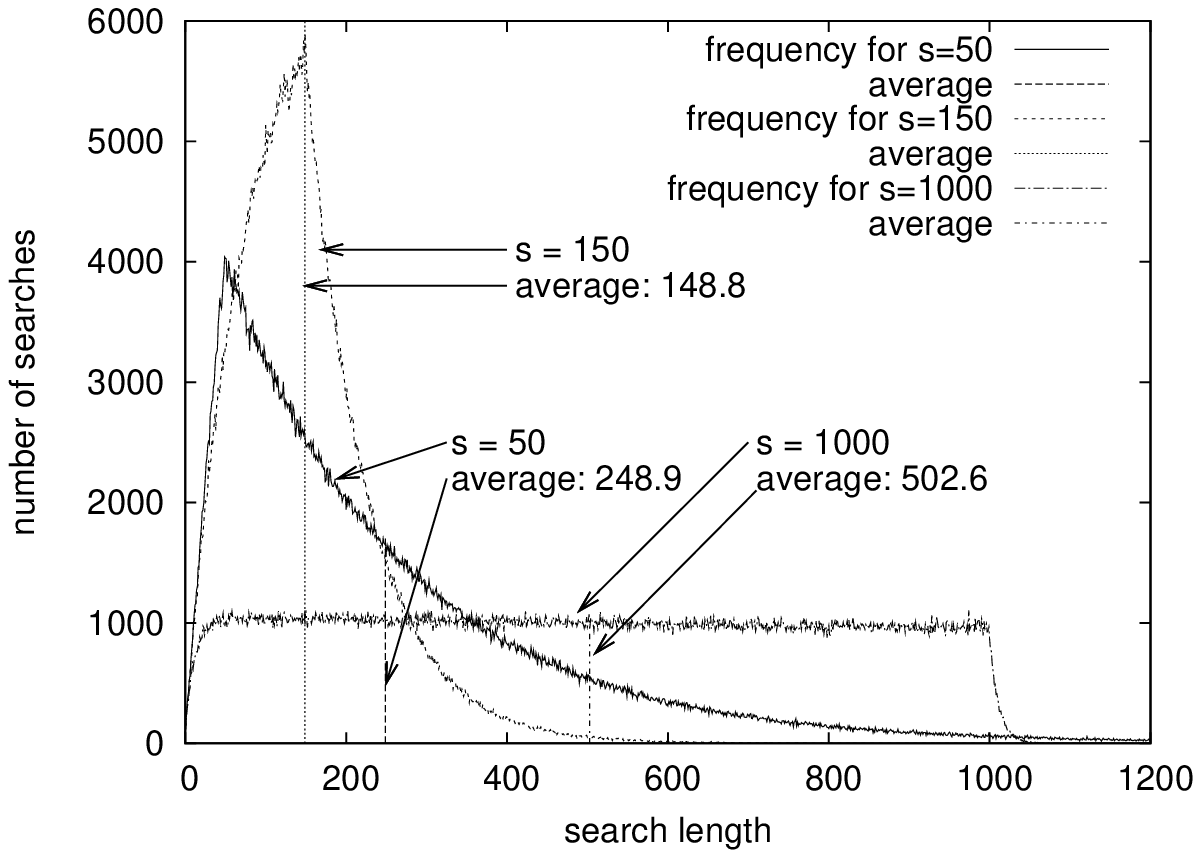}
 \label{fig:ldistr_s10_150_1000_reg} }
\hfil 
 \subfloat[Search lengths for \sopt\ and for $\p=0,0.01,0.1$.]{
 \includegraphics[width=7.8cm]{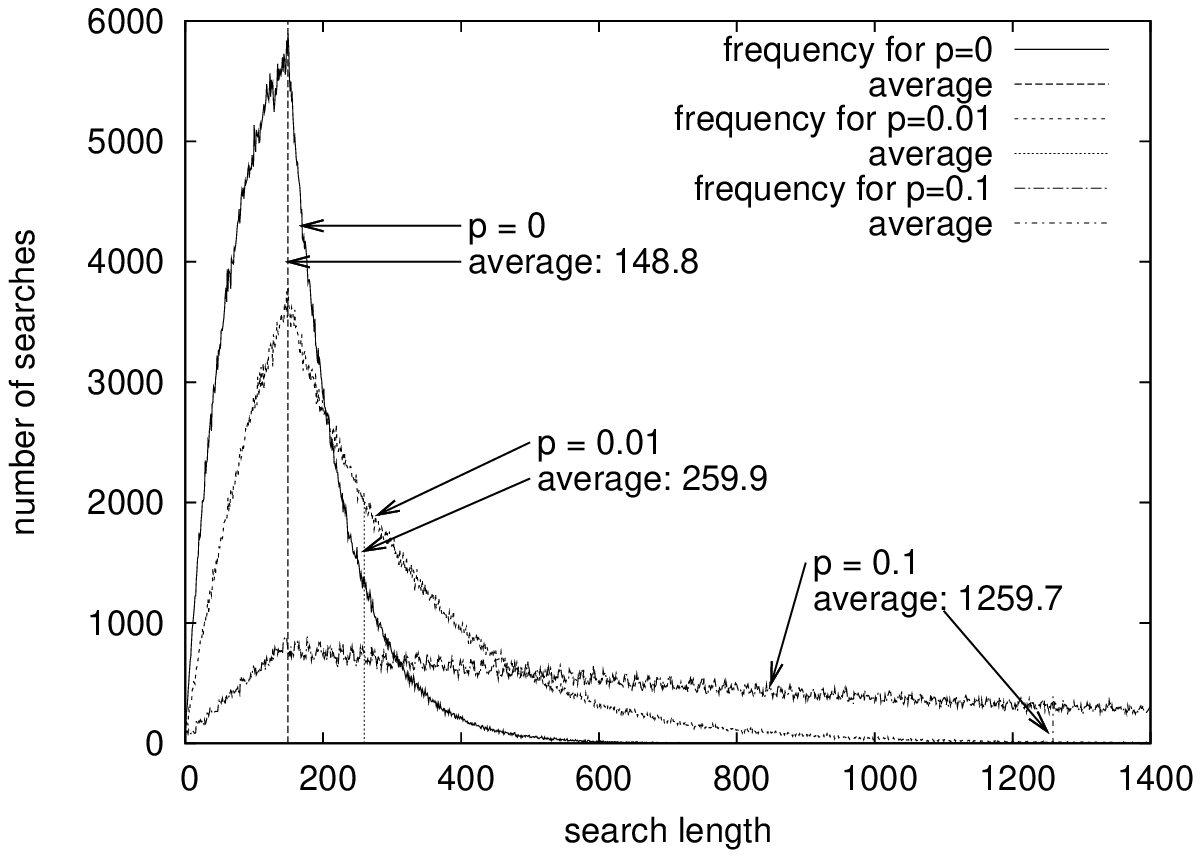}
 \label{fig:ldistr_all_sopt_pall} }
}

 \caption{Distributions of search lengths (histograms) with PWs that are not reused in the regular network.}
 \label{fig:ldistr_reg}
\end{figure*} 



\paragraph{Length distributions.} 
\smallskip

We first obtain the lengths distributions of searches  using PWs that are never reused. Later in this section we will discuss the effect of having a limited number of partial random walks that are reused.
We consider each random walk to be the total walk of a search based on PWs. For each original random walk, we break it in pieces of size \s, which are taken as the PWs that make up the total walk. Then we consider a search that uses those PWs and count the number of hops (jumps plus trailing steps plus unnecessary steps). This gives the length of the search if it had been constructed using those (precomputed) PWs. Note that the PWs are not reused because they are obtained from independent (real) random walks.

The search length distributions in the regular network for $\p=0$ and for several values of \s\ are shown in Figure~\ref{fig:ldistr_reg}\subref{fig:ldistr_s10_150_1000_reg}. The plots also show, as vertical bars, the average search lengths computed from each distribution. These average values are very close to the expected values calculated with Equation~\ref{eq:Lsexp_p} ($\Lexp_{50} = 248.9$, $\Lexp_{150} = 149.0$ and $\Lexp_{1000} = 510.2$).
Therefore, our model accurately predicts average lengths of searches based on PWs of size \s\ in the three types of networks considered in our experiments.

As for the shape of the distributions, we observe that for low \s\ ($\s=50$ in Figure~\ref{fig:ldistr_reg}\subref{fig:ldistr_s10_150_1000_reg}) the search lengths are dominated by the number of jumps, which is proportional to the length of the total walk. On the other hand, for high \s\ ($\s=1000$ in Figure~\ref{fig:ldistr_reg}\subref{fig:ldistr_s10_150_1000_reg}) the distribution adopts a rather uniform shape. Search lengths are dominated here by the number of trailing steps in the last PW, and this has approximately an uniform distribution between 0 and $\s-1$, as mentioned earlier. The optimal length for the PWs, \sopt\ ($\s=150$ in Figure~\ref{fig:ldistr_reg}\subref{fig:ldistr_s10_150_1000_reg}), represents a transition point between these two effects. The shape is such that the values around the average search length (which approximately equals \sopt, according to Equation~\ref{eq:lopt}) are also the most frequent.


Once it has been found the optimal length for the PWs \sopt\ (which is known to be independent of the value of $p$), we investigate the effect of the probability of false positive of Bloom filters in these distributions. Figure~\ref{fig:ldistr_reg}\subref{fig:ldistr_all_sopt_pall} shows the distributions of search lengths (histograms) for the regular network when $\s=\sopt$ and for several values of \p. 
It can be seen that the distributions get wider and lower as \p\ grows, pushing average search lengths to higher values, in accordance with Figure~\ref{fig:figA}\subref{fig:lopt_p}. However, we observe that the most frequent lengths remain the same regardless of the value of \p. For $p=0$, the most frequent value for each network approximately equals the average search length which, in turn, approximately equals the optimal length of the PWs ($\sopt=150$ for the regular network). For greater values of \p, the average search length grows while the most frequent value stays the same.

Regarding the distributions for the ER and the scale-free networks, they have similar shapes and are not shown here. However, we have used these distributions to obtain Table~\ref{tab:H_reductions}\subref{tab:H_reduction_RW_PWRW} (explained below).

\paragraph{Effect of reusing PWs.}
\smallskip

At this point, we note that we have been assuming that PWs are never reused. However, in practical scenarios it seems quite reasonable to consider a limited number of partial random walks that are reused. In Appendix~\ref{search_length_distr_reused_PW} we have explored the distributions of search lengths when the total walks are built reusing a limited number $w$ of PWs precomputed in each node. As it can be readily seen there, we conclude that, for the types of networks in our experiment, just two precomputed PWs per node are enough to obtain searches whose lengths are statistically similar to those that would be obtained with PWs that are not reused. So, we can say that our results using not reused PWs are also valid when using a limited number of PWs that are reused.

\subsubsection{Comparison of performance with respect to random searches.} 
\smallskip

Finally, in Table~\ref{tab:H_reductions}\subref{tab:H_reduction_RW_PWRW} we compare the performance of the proposed search mechanism with respect to random walk searches. We can see that the reduction in the average search length that PW-RW achieves with respect to simple random walk is lower for higher $p$, ranging from around $98\%$ in the case when $p=0$ to $88\%$ when $p=0.1$. Furthermore, we also see that the achieved reductions are independent of the network type.

\begin{table}
\centering
\subfloat[PW-RW with respect to random walk searches]{
\begin{tabular}{lccc} \hline
 & \multicolumn{3}{c}{\rule{0pt}{11pt}Reduction of $\lave$ (\%)} \\ \cline{2-4}
\rule{0pt}{11pt}Network type & $\p=0$ & $\p=0.01$ & $\p=0.1$ \\ \hline 
Regular      & 98.67 & 97.68 & 88.73 \\
ER	     & 98.71 & 97.68 & 88.42 \\
Scale-free   & 98.83 & 97.79 & 88.43 \\ 
\hline
\end{tabular}
\label{tab:H_reduction_RW_PWRW}
} \hfil
\subfloat[PW-SAW with respect to PW-RW]{
\begin{tabular}{lccc} \hline
 & \multicolumn{3}{c}{\rule{0pt}{11pt}Reduction of $\lave$ (\%)} \\ \cline{2-4}
\rule{0pt}{11pt}Network type & $\p=0$ & $\p=0.01$ & $\p=0.1$ \\ \hline 
Regular      & 5.67 & 8.22 & 11.24 \\
ER	     & 6.25 & 9.10 & 11.88 \\
Scale-free   & 6.53 & 9.75 & 12.65 \\ 
\hline
\end{tabular}
\label{tab:H_reduction_PWRW_PWSAW}
}
\caption{Reductions of average search lengths.}
\label{tab:H_reductions}
\end{table}

\section{Choose-First PW-SAW}
\label{sec:lave_SAW}

As it was  pointed in Section~\ref{sec:analytical} when we introduced the PW construction mechanism in Stage 1, a possible variation consists of using self-avoiding walks (SAW) instead of simple random walks. The resulting search mechanism is called PW-SAW. The basic idea is to revisit less nodes, thus increasing the chances of locating the desired resource. 
In short, a SAW chooses the next node to visit uniformly at random among the neighbors that have not been visited so far by the walk. If all neighbors have already been visited, it chooses uniformly at random among all neighbors, like a simple random walk.

\paragraph{Analysis of Choose-First PW-SAW.}
\label{sec:analysis_PWSAW}

When PWs are self-avoiding walks, their concatenation is not a random walk, and hence Theorem~\ref{t-expectedL} is no longer valid. 
We state a new theorem here for the choose-first PW-SAW mechanism and prove it in Appendix~\ref{s-proofofthm-SAW} using a different approach.

\begin{theorem}
\label{t-expectedL-SAW}
If the expected number of trailing steps is assumed to be uniformly distributed in $[0, s-1]$, then the expected search length of PW-SAW is 
\begin{equation}
 \Lsexp = \frac{1}{N}\,\sum_k \nk\,\left(\frac{1}{\ps(k)} \cdot \left(\pn(k) + s\cdot\pf(k)\right) + \frac{s-1}{2}\right). 
\end{equation}
\end{theorem}
In the above theorem, \pn, \ps, and \pf\ are the probabilities that the query of the Bloom filter of the chosen PW in the current node returns a (true) negative, a true positive, and a false positive result, respectively, as a funcion of $k$, the degree of the node holding the resource. The proof in Appendix~\ref{s-proofofthm-SAW} gives expressions for these probabilities.

\paragraph{Expected Search Length in PW-SAW.}

In this section, we compare the analytic results from the model with experimental data from simulations. Figure~\ref{fig:figB}\subref{fig:ls_SAW_all_p0_s} shows the expected search length (\Lsexp) as a function of the size of PWs (\s) in a regular network, an ER network and a scale-free network, for $p=0$. The curves in this graph are plotted using Equation~\ref{eq:Lsexp_SAW3} and previous equations. 

According to the results computed using the PW-SAW model, the minimum search lengths occur for values around $\s=141$, $\s=149$ and $\s=167$ for the regular, ER and scale-free networks, respectively. These values are slightly lower than the ones predicted by the PW-RW model (Figure~\ref{fig:figA}\subref{fig:ls_s}), which were $\sopt=150, 157$ and $174$, respectively. 

Both the model curves and the simulation experiments have been computed for $\w=5$, chosen as a reference value. However, it has been observed that very similar results are obtained if we change the value of \w. Furthermore, plots of the model equations for different values of \w\ are coincident. This behavior was also observed for PW-RW (Section~\ref{sec:search_length_distr}), where we found that the average search length remained almost constant as we increased \w. The reason for this is that the probability of the resource being in the chosen PW (\pres\ in Equation~\ref{eq:Pij}) does not depend on the number of PWs in the node.

\begin{figure*}[!t]
 \centerline{
  \subfloat[As a function of \s\ for $\p=0$. The optimal points $(\sopt,\Lexpopt)$ for each network are $(141,139.92)$, $(149,148.55)$, and $(167,164.75)$.]{
   \includegraphics[width=7.8cm]{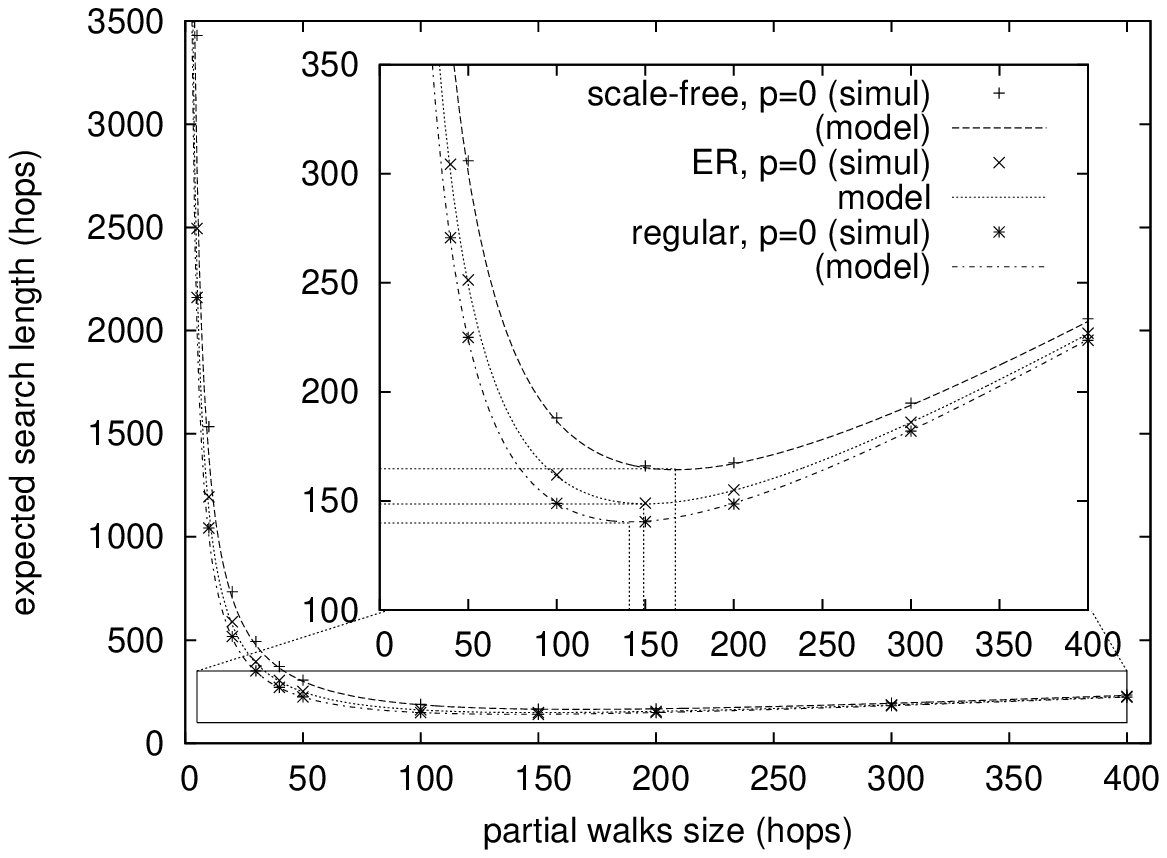} \label{fig:ls_SAW_all_p0_s}
  }
  \hfil
  \subfloat[Comparison with PW-RW for $p=0,0.01,0.1$.]{
   \includegraphics[width=7.8cm]{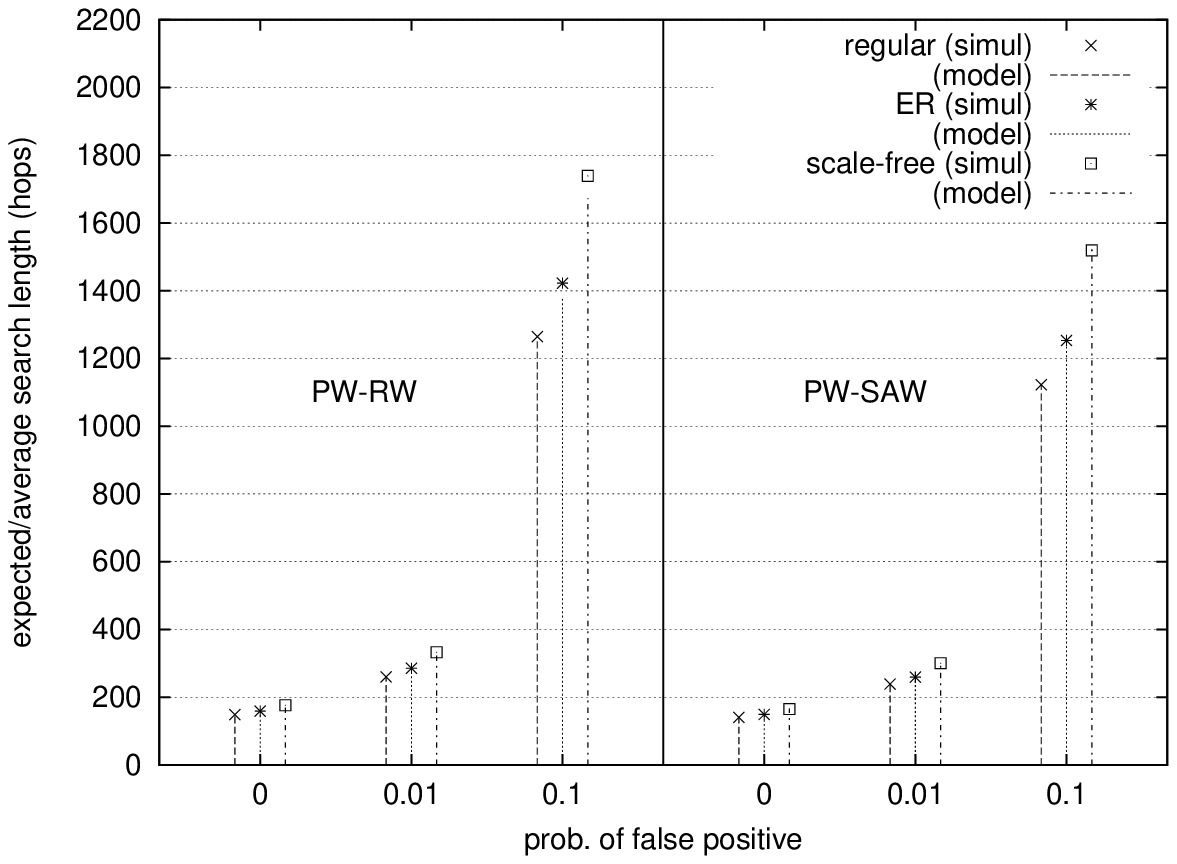} \label{fig:l_RW_SAW_sm}
  }
 } 
 \caption{Expected search length of PW-SAW in a regular network, an ER network and a scale-free network. }
 \label{fig:figB}
\end{figure*}



We now compare the results of the PW-RW and PW-SAW mechanisms. Figure~\ref{fig:figB}\subref{fig:l_RW_SAW_sm} shows results for PW-RW (left part) and for PW-SAW (right part), in the three networks considered in our study, and for values of $\p=0, 0.01$ and $0.1$. Expected search lengths from the analytical models are shown as vertical bars, while average search lengths from the simulations experiments are shown as points. The size of the PWs has been set to $\s=150, 157$ and $174$ for the regular, ER and scale-free networks, respectively, which are the optimal values predicted by the PW-RW model. For all the networks, we have found a very good correspondence between model predictions and simulation results.


\paragraph{Comparison of performance with respect to choose-first PW-RW.} 
\smallskip

If we compare the performance of the proposed search mechanisms, we observe that the reduction in the average search length that PW-SAW achieves with respect to PW-RW for a given \p\ is largest for the scale-free network, followed by the ER network and then by the regular network. For each network type, the reduction is larger for higher \p. Actual values can be found in Table~\ref{tab:H_reductions}\subref{tab:H_reduction_PWRW_PWSAW}.



\section{Check-First PW-RW and PW-SAW}
\label{sec:check-first}

We now present the check-first versions of the PW-RW and PW-SAW search mechanisms, introduced in Section~\ref{sec:analytical}.
Suppose the search is currently in a node and it needs to pick one of the PWs in that node to decide whether to traverse it or to jump over it.
With the new check-first mechanism,
it first \emph{checks} the associated resource information of \emph{all} the PWs of the node, and then randomly \emph{chooses} among the PWs with a positive result, if any (otherwise, it chooses among all PWs of the node, as the choose-first version).
These check-first mechanisms improve the performance of their choose-first counterparts, since the probability of choosing a PW with the resource increases. This comes at the expense of slightly incrementing the processing power used since several PWs need to be checked, but without incurring extra storage space costs. 

A minor additional
difference between the algorithms is that
in the check-first version, the resource information is registered from the \emph{first} node (the node next to the current node) to the \emph{last} node in the PW. This change slightly improves the performance of the new version, since the probability of choosing a PW with the resource increases also in the cases where the resource is held by the last node of the PW.
We have adapted the analysis presented in Section~\ref{sec:analysis_PWSAW} to reflect the new behavior of the check-first PW-RW/PW-SAW mechanisms. Details can be found in Appendix~\ref{sec:analysis_check-first}.


\paragraph{Expected Search Length in Check-First PW-RW/PW-SAW.}

Figure~\ref{fig:ls_all_p0.01_s} shows the expected search length (\Lsexp) as a function of the size of PWs (\s) in a regular network for the four mechanisms presented so far: choose-first PW-RW/PW-SAW, and check-first PW-RW/PW-SAW, for $p=0.01$ and $w=5$. We observe that the check-first mechanisms achieve a lower minimum expected search length than the original choose-first mechanisms, as expected. In fact, the expected search length can be lowered further by increasing $w$, the number of PWs per node, clearly at the expense of increasing the cost of the PWs construction stage. Also interesting is the observation that the minimum expected search length occurs for significantly lower \s\ (\sopt\ falls from 150 to about 50), meaning shorter PWs in the nodes, which in turn decreases the cost of the PWs construction stage.  
With regard to the PW-SAW mechanisms, we note that they achieve a slight decrease in the expected search length with respect to the PW-RW mechanisms, for the check-first version as well as for the choose-first version (which was already observed in Table~\ref{tab:H_reductions}). Results for the ER and scale-free networks are similar and are omitted here.

\begin{figure}
 \centering
 \includegraphics[width=7.8cm]{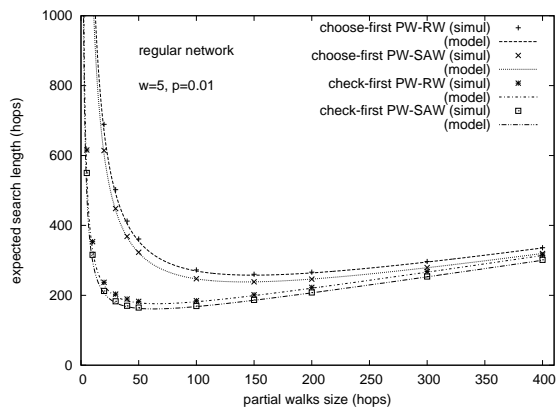}
 \caption{Expected search length of choose-first and check-first versions of PW-RW and PW-SAW as a function of \s\ in a regular network for $\p=0.01$ and $w=5$. Simulation and model results.}
 \label{fig:ls_all_p0.01_s}
\end{figure}


\section{Future Work}
The proposed resource location mechanisms could be improved with new strategies to choose from the PWs available at the nodes. Smarter (and more costly) variants of RWs could be used as PWs. It would be interesting to compare their application to unstructured P2P networks with algorithms for structured overlays like DHT or quorum systems.

\bibliographystyle{unsrt}
\bibliography{randomwalks,networks}





\newpage
\appendix


\section{Distributions of the Number of Trailing Steps}
\label{ssec:distr_tsteps}

The previous proof of Theorem~\ref{t-expectedL} 
assumes that the distribution of the number of trailing steps in the last partial walk until the search finds the resource is uniform between 0 and $\s-1$, corresponding to the cases where the first node/last node in the partial walk holds the desired resource. Recall that the Bloom filter stores the resources held by the \s\ first nodes in the partial walk, from the node that precomputed the partial walk to the one before its last node (which is included in the partial walks departing from it).
We have obtained that distribution from the $10^6$ searches in our experiment for each of the three networks. 
Figure~\ref{fig:tdistr_reg} shows
the distributions for the regular network when $\s=10$, $\s=\sopt=150$ and $\s=1000$. Distributions for the ER and scale-free networks are similar in shape.

\begin{figure}
 \centering
 \includegraphics[width=8.5cm]{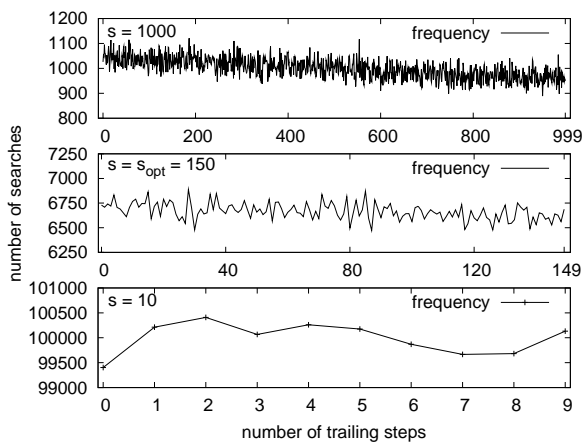}
 \caption{Distributions of the number of trailing steps in the regular network.}
 \label{fig:tdistr_reg}
\end{figure}

It is observed that there is a slight decrease on the frequency as the number of steps grows. This is due to the fact that the number of trailing steps is essentially the length of the total walk modulus the length of partial walks (\s).  The total walk is a random walk, and its distribution can be obtained approximately by Equation~\ref{eq:approx_formula}.\footnote{The distribution of simple random walk searches has also been obtained experimentally, showing that Equation~\ref{eq:approx_formula} is a good approximation.} Since it is a decreasing function, as it is shown below, the frequency on the left end of an interval of width \s\ is always higher than the frequency on the right end, thus accounting for the observed decrease. 

This means that the result provided by Theorem~\ref{t-expectedL} 
is \emph{pessimistic}, since the estimated average number of trailing steps is slightly higher than the real one. Results in Section~\ref{sec:perf_eval} have shown that values of average search lengths predicted by Equation~\ref{eq:Lsexp_p} are very similar to values computed from simulations, with larger error for higher values of \s.


The probability distribution of simple random walk searches can be estimated using Equation~\ref{eq:approx_formula}. It can be demonstrated that it is strictly decreasing, that is: $P_{i} - P_{i-1} < 0$ for $0 \leq i < \infty$, as follows:
\begin{eqnarray*}
 P_0 & = & \frac{1}{N}, \\
 P_i & = & \left(1 - \sum_{j=0}^{i-1}{P_j}\right) \cdot \frac{1}{N-1},\ \mathrm{for}\ i>0.
 \label{eq:approx_formula}
\end{eqnarray*}

First, it is shown by induction that $0 < \sum_{i=0}^k P_i  < 1$ for $k\geq 0$ and $N>0$. It hols trivially for $k=0$. Then, it is also true for $k>0$ if it holds for $k-1$:
\begin{eqnarray*}
 \sum_{i=0}^k P_i & = & \sum_{i=0}^{k-1} P_i + \left( 1 - \sum_{i=0}^{k-1} P_i  \right)\cdot \frac{1}{N-1}\\
                  & = & \frac{N-2}{N-1} \cdot \sum_{i=0}^{k-1} P_i + \frac{1}{N-1} \\
		  & < & \frac{N-2}{N-1} + \frac{1}{N-1} = 1.
\end{eqnarray*}

Next, it is shown that $0 < P_i < 1$ for $i\geq 0$ as a corollary of the previous result. It is checked for $i=0$ by inspection. For $i>0$, we have that $P_i = \left ( 1 - \sum_{j=0}^{i-1} P_j \right )\cdot\frac{1}{N-1} $. By the previous result:
\begin{eqnarray*}
 0 < 1 - \sum_{j=0}^{i-1} P_j < 1,
\end{eqnarray*}
then we have that:
\begin{eqnarray*}
 0 < P_i = \left(1 - \sum_{j=0}^{i-1} P_j\right)\cdot\frac{1}{N-1} < 1.
\end{eqnarray*}

Finally, it is shown that $P_i - P_{i-1} < 0$ for $i>0$. For $i=1$, it is checked by inspection. For $i>1$:
\begin{eqnarray*}
 P_i - P_{i-1}\!\!\! &=& \!\!\!\left( 1 - \sum_{j=0}^{i-1} P_j \right)\frac{1}{N-1} - \left( 1 - \sum_{j=0}^{i-2} P_j \right)\frac{1}{N-1} \\
              \!\!\! &=& \!\!\! -\frac{P_{i-1}}{N-1}.
\end{eqnarray*}
Since we have shown that $0 < P_{i-1} < 1$, it follows that $P_i - P_{i-1} < 0$.

\section{Expectation of a Random Variable with a Binomial Distribution in Which the Number of Experiments is Another Random Variable} 
\label{sec:demo_exp}

Let $X$ be a random variable with sample space $S = \mathbb{N}_0 = \{0,1,2\ldots\}$. Let $Y$ be a random variable representing the number of successes when $X$ experiments are performed with a success probability $p$. $Y$ has a binomial probability distribution $Y\sim \mathrm{B}(X,p)$, where the number of experiments is, in turn, a random variable. Then, from the definition of expectation and applying the Total Probability Theorem, the expectation of $Y$ is $\mathrm{E}[Y] = \mathrm{E}[X]\cdot p$.
\begin{eqnarray*}
\Exp{Y} & = & \sum_{y=0}^\infty y \cdot \Prob{Y=y} \\
  & = & \sum_{y=0}^\infty y \cdot \left\{ \sum_{x=0}^\infty \Prob{Y=y|X=x} \cdot \Prob{X=x} \right\}  \\
  & = & \sum_{x=0}^\infty \Exp{Y|X=x} \cdot \Prob{X=x} \\
  & = & \sum_{x=0}^\infty x \cdot p \cdot \Prob{X=x} = \Exp{X} \cdot p.
\end{eqnarray*}


\section{Proof of Theorem~\ref{t-expectedL-SAW}}
\label{s-proofofthm-SAW}

\begin{proof}
We write a recurrence equation for the expected length, given that the search is currently in any of the nodes it visits. Since we have defined the expected search length for any pair of source and target nodes, the expected length of the search from the current node and the expected length of the search from the source node are the same. Denoting it by \Lsexp, as in the previous section, we can write:
\begin{equation}
 \Lsexp = (\Lsexp + 1)\cdot \pn + (\Lsexp + s)\cdot \pf + \frac{s-1}{2}\cdot \ps,
 \label{eq:Lsexp_SAW_recur}
\end{equation}

\noindent 
where \pn, \ps, and \pf\ are the probabilities that the query of the Bloom filter of the chosen partial walk in the current node returns a (true) negative, a true positive, and a false positive result, respectively, with $\pn+\ps+\pf=1$.
Solving for \Lsexp, we obtain:
\begin{equation}
 \Lsexp = \frac{1}{\ps} \cdot (\pn + s\cdot\pf) + \frac{s-1}{2}.
 \label{eq:Lsexp_SAW}
\end{equation}

\noindent This equation can be rewritten as:
\begin{equation}
 \Lsexp = \frac{1-\ps}{\ps} \cdot \left(\frac{\pn}{1-\ps} + s\cdot\frac{\pf}{1-\ps}\right) + \frac{s-1}{2}
 \label{eq:Lsexp_SAW2},
\end{equation}

\noindent 
which is an alternative formulation of the expected search length, in terms of the expected number of partial walks of the search (\Pexp, as defined in Section~\ref{sec:lave_RW}). Note that \mbox{$(1-\ps)/\ps$} is the expectation of \Pexp, a geometric random variable representing the number of failures before a Bloom filter returns a true positive (with probability \ps). The fractions within the parenthesis are, respectively, the probabilities of jumping a partial walk or traversing it, conditional on the fact that the Bloom filter does not return a true positive. Therefore, the terms in the parenthesis are the expectations of \Jexp\ and \Uexp, binomial random variables representing the number of jumps and the number of partial walks that are unnecessarily traversed, respectively, as defined in Section~\ref{sec:lave_RW}.

We now calculate the probabilities in the equations above using $P(i,j)$, the probability that, in the \w\ partial walks of a node, there are $i$ partial walks that contain the node that holds the resource (i.e., their Bloom filters return a true positive), and $j$ partial walks that do not contain the resource, but whose filters return false positives:
\begin{equation}
 P(i,j) = B(\w,\pres,i)\cdot B(\w-i,p,j),
 \label{eq:Pij}
\end{equation}

\noindent
where $B(m,q,n)$ is the coefficient of the binomial distribution: 
$B(m,q,n) = \left(\begin{array}{c} m\\ n \end{array}\right)\cdot q^n \cdot (1-q)^{(m-n)}$.

In Equation~\ref{eq:Pij} we are using \pres, defined as the probability that a partial walk includes the node that holds the desired resource. This probability is proportional to the degree of the node that holds the resource, since the probability that a random walk visits a node depends on its degree (see~\cite{rw:Lovasz93}, for example). We assume known the number of nodes of each degree $k$ in the network, i.e., its degree distribution, which we denote by \nk.

Denoting by $k$ the degree of the node that holds the resource, the probability that a partial walk of size \s\ contains the resource is then $\pres(k)$, and it can be estimated as:
\begin{equation}
 \pres(k) = 1 - \prod_{l=0}^{s-1} \left(1 - \frac{k}{S-l\kave}\right),
 \label{eq:pres}
\end{equation}
where $S$ denotes the number of endpoints in the network ($S=\sum_k k\,\nk$) and \kave\ denotes the average degree of the network ($\kave=\sum_k k\,\nk/N$). Each factor in the product in Equation~\ref{eq:pres} represents the probability that the resource is not found in the $l$th hop of a partial walk, conditional on the fact that it was not found in the previous hops of that partial walk. Note that the fraction $k/(S-l\kave)$ is the probability of the $l$th hop finding the resource, expressed as the number of endpoints that belong to the node that holds the resource divided by the total number of endpoints in the network, except those belonging to nodes already visited by the partial walk, which are \kave\ per hop, on the average. 


Now we rewrite Equation~\ref{eq:Pij} making its dependence on $k$ explicit:
\begin{equation}
 P(i,j|k) = B(\w,\pres(k),i)\cdot B(\w-i,p,j),
\end{equation}
Then, the probabilities in Equations~\ref{eq:Lsexp_SAW_recur}~and~\ref{eq:Lsexp_SAW} are:
\begin{eqnarray}
  \ps(k) & = & \sum_{i=1}^{w} \sum_{j=0}^{w-i} P(i,j|k)\cdot \frac{i}{w} \nonumber \\
  \pf(k) & = & \sum_{i=0}^{w} \sum_{j=1}^{w-i} P(i,j|k)\cdot \frac{j}{w} \nonumber \\
  \pn(k) & = & 1-\ps(k)-\pn(k).
 \label{eq:SAW_probs}
\end{eqnarray}

The expected search length can be finally obtained weighing Equation~\ref{eq:Lsexp_SAW} with the probability that the resource is in a node with degree $k$, which is $\nk/N$, for all values of $k$:
\begin{equation}
 \Lsexp = \frac{1}{N}\,\sum_k \nk\,\left(\frac{1}{\ps(k)} \cdot \left(\pn(k) + s\cdot\pf(k)\right) + \frac{s-1}{2}\right).
\label{eq:Lsexp_SAW3}
\end{equation}
\end{proof}

\section{Alternative Analysis for Choose-First PW-RW}
\label{sec:alt-analysis-PWRW1}

This section presents an alternative analysis for the model of the choose-first PW-RW mechanism described in Section~\ref{sec:lave_RW}. This analysis is based on that of the PW-SAW mechanism, presented in Section~\ref{sec:lave_SAW} and proved in Appendix~\ref{s-proofofthm-SAW}. In fact, only the expression for $\pres(k)$ (Equation~\ref{eq:pres}), defined as the probability that a given PW contains the node that holds the resource, needs to be rewritten to reflect the fact that the PW is a simple random walk instead of a self-avoiding random walk. The new expression is:


\begin{equation}
 \pres(k) = 1 - \left(1 - \frac{k}{S-\kave_{rw}}\cdot\frac{\kave_{rw}-1}{\kave_{rw}}\right)^s.
 \label{eq:pres_1}
\end{equation}
The first fraction within the parenthesis in Equation~\ref{eq:pres_1} is the ratio of positive endpoints (the degree of the node that holds the resource) and all endpoints in the network ($S=\sum_k k\,\nk$) except those of the current node. We use $\kave_{rw}$, which denotes the expectation of the degree of a node visited by a random walk, as an estimation of the degree of the current node. It can be obtained as:
\begin{equation}
 \kave_{rw} = \sum_k k\cdot \frac{k\cdot n_k}{S} = \frac{1}{S} \cdot \sum_k k^2\cdot n_k.
 \label{eq:pkexp_1}
\end{equation}
The second fraction within the parenthesis in Equation~\ref{eq:pres_1} corrects the previous ratio taking into account that, when at a node of a given degree, the probability of not going backwards (and therefore having the chance to find the resource) is the probability of selecting any of its endpoints but the one that connects it with the node just visited.

The rest of the equations in Appendix~\ref{s-proofofthm-SAW} are valid for this analysis of the choose-first PW-RW mechanism.

\omt{

Now we rewrite Equation~\ref{eq:Pij_1} making its dependence on $k$ explicit:
\begin{equation}
 P(i,j|k) = B(\w,\pres(k),i)\cdot B(\w-i,p,j),
\end{equation}
Then, the probabilities in Equations~\ref{eq:Lsexp_RW_recur_1}~and~\ref{eq:Lsexp_RW_1} are:
\begin{eqnarray}
  \ps(k) & = & \sum_{i=1}^{w} \sum_{j=0}^{w-i} P(i,j|k)\cdot \frac{i}{w} \nonumber \\
  \pf(k) & = & \sum_{i=0}^{w} \sum_{j=1}^{w-i} P(i,j|k)\cdot \frac{j}{w} \nonumber \\
  \pn(k) & = & 1-\ps(k)-\pn(k).
 \label{eq:RW_probs_1}
\end{eqnarray}

The expected search length can be finally obtained weighing Equation~\ref{eq:Lsexp_RW_1} with the probability that the resource is in a node with degree $k$, which is $\nk/N$, for all values of $k$:
\begin{equation}
 \Lsexp = \frac{1}{N}\,\sum_k \nk\,\left(\frac{1}{\ps(k)} \cdot \left(\pn(k) + s\cdot\pf(k)\right) + \frac{s-1}{2}\right).
\label{eq:Lsexp_RW3_1}
\end{equation}
} 


\section{Analyses of Check-First PW-RW and PW-SAW}
\label{sec:analysis_check-first}

This section presents the analyses of the check-first versions of the PW-RW and PW-SAW mechanisms introduced in Section~\ref{sec:check-first}. This analysis is based on the analysis of the choose-first versions of the mechanisms (presented for PW-SAW in Section~\ref{sec:analysis_PWSAW} and adapted for PW-RW in Appendix~\ref{sec:alt-analysis-PWRW1}).

Most of the expressions in the analysis of the choose-first versions are still valid for the check-first versions of the mechanisms, so we present here only the equations that need to be modified to reflect the new behavior. That is the case of Equations~\ref{eq:SAW_probs} for the probabilities of choosing a PW with a true positive, false positive, and negative result, respectively. Their counterparts follow. Remember that $i$ and $j$ represent the number of PWs of the node that return a true positive result and an false positive result, respectively:
\begin{eqnarray}
  \ps & = & \sum_{i=1}^{w} \sum_{j=0}^{w-i} P(i,j)\cdot \frac{i}{i+j}, \nonumber \\
  \pf & = & \sum_{i=0}^{w-1} \sum_{j=1}^{w-i} P(i,j)\cdot \frac{j}{i+j}, \nonumber \\
  \pn & = & P(0,0) = 1-\ps-\pf.
 \label{eq:RW_probs_checkfirst}
\end{eqnarray}

The expression for $\pres(k)$\ in Equation~\ref{eq:pres_1} is still valid for check-first PW-RW. However, Equation~\ref{eq:pres} needs to be modified for check-first PW-SAW, since the range of nodes whose resources are associated with the PW has changed from $[0,s-1]$ to $[1,s]$:

\begin{equation}
 \pres(k) = 1 - \prod_{l=1}^{s} \left(1 - \frac{k}{S-l\kave}\right).
 \label{eq:pres_SAW_checkfirst}
\end{equation}
Finally, Equation~\ref{eq:Lsexp_SAW3} also needs modification (for chech-first PW-SAW) in the expectation of trailing steps, for the same reason. The new version, which completes the analysis of the check-first mechanisms, is:
\begin{equation}
 \Lsexp = \frac{1}{N}\,\sum_k \nk\,\left(\frac{1}{\ps(k)} \cdot \left(\pn(k) + s\cdot\pf(k)\right) + \frac{s}{2}\right).
\label{eq:Lsexp_RW3_checkfirst}
\end{equation}

\section{Searches based on reused partial walks}
\label{search_length_distr_reused_PW}

In this section, we explore the distributions when the total walks are built reusing a limited number \w\ of partial walks precomputed in each node. This is in contrast with our initial assumption that precomputed partial walks are not reused in searches. Here, we attempt to answer the question ``How many partial walks does a node need to precompute, for the search lengths distribution to be similar to that corresponding to never reusing partial walks?''. Our results show that, for the networks considered in our experiment, and for the optimal partial walk size (\sopt), it is enough to have as few as \emph{two} precomputed partial walks in every node. The extreme case of having just \emph{one} precomputed partial walk yields a significant fraction of unfinished searches, since it is relatively easy to build walks that are loops that do not visit all the nodes. Indeed, if the last node of a partial walk is a node whose (only) partial walk has been previously used in that total walk, it will take the search to the same place again, resulting in a never-ending loop. However, if a node has several partial walks, and the search chooses one randomly among them (for the next jump or partial walk traversal), the chances of entering a loop are very small.

Figures~\ref{fig:L_distr_reg}\subref{fig:ldistr_reg10_s150_p0_PWinf_1_2}~to~\ref{fig:L_distr_reg}\subref{fig:ldistr_reg10_s150_p0.1_PWinf_1_2} show the search lengths distributions in the regular network. The top plots of these figures show the length distributions of searches based on PWs that are not reused. The middle and bottom plots show the length distributions of searches based on reusing a single partial walk or two partial walks per node, respectively.

We note that the shape of the distributions is the same for all values of \w. However, distributions for $\w=1$ are lower, and the average search length (marked as a vertical bar) is also smaller. This is due to a significant percentage of unfinished searches (about 26.3\%), left out of the histograms, due to loops as explained above. 
If we focus now on the distributions for $\w=2$, we observe that both the distribution and the average search length are very similar to those for PWs that are not reused. We have performed additional experiments with higher values of \w, confirming this observation. This suggests that just two precomputed partial walks per node are enough to obtain a behavior close to the theorical case of using PWs that are not reused. 
The distributions of searches in the ER network and the scale-free network are omitted here, since their shape and the conclusions drawn are the same as for the regular network. 


\newcommand{\ancho}{5.4}

\begin{figure*}
 \centering
 \subfloat[$p=0$.]{
  \includegraphics[width=\ancho cm]{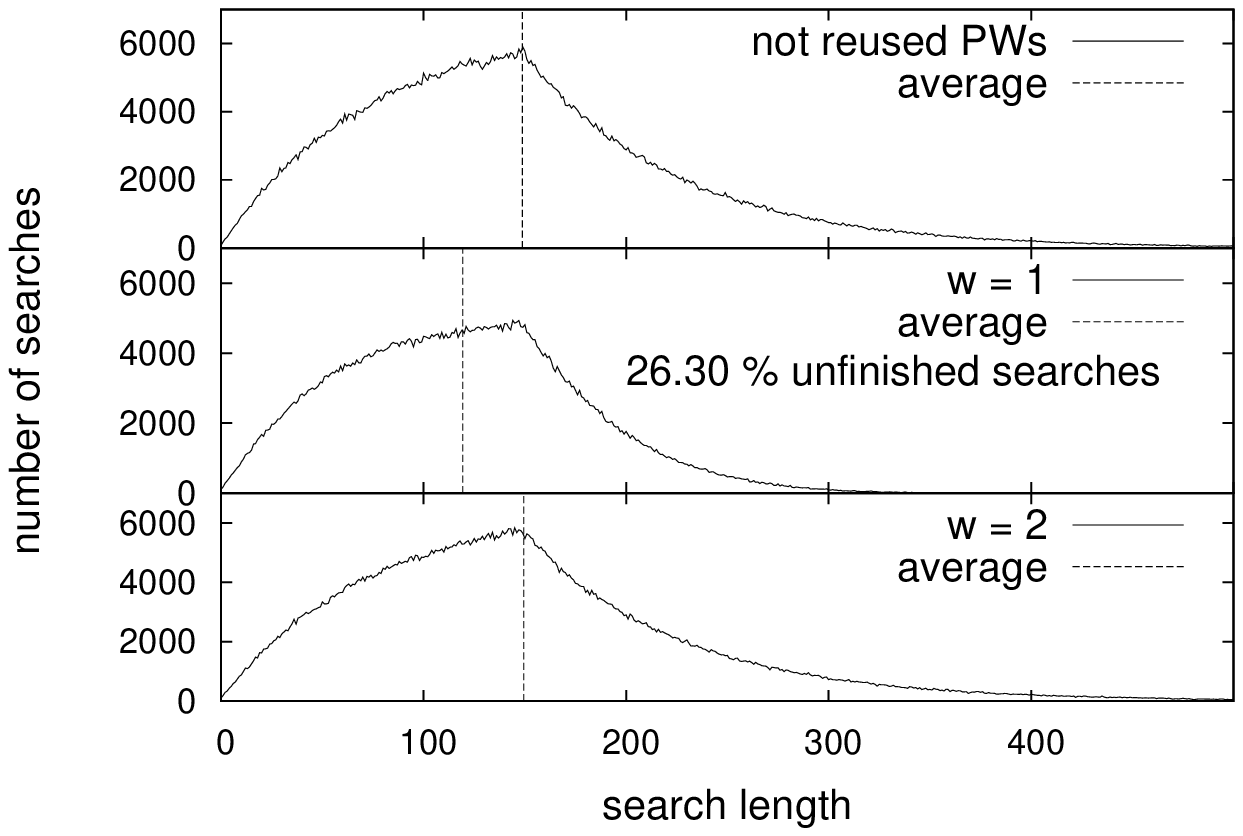}
  \label{fig:ldistr_reg10_s150_p0_PWinf_1_2}
 }
 \subfloat[$p=0.01$.]{
  \includegraphics[width=\ancho cm]{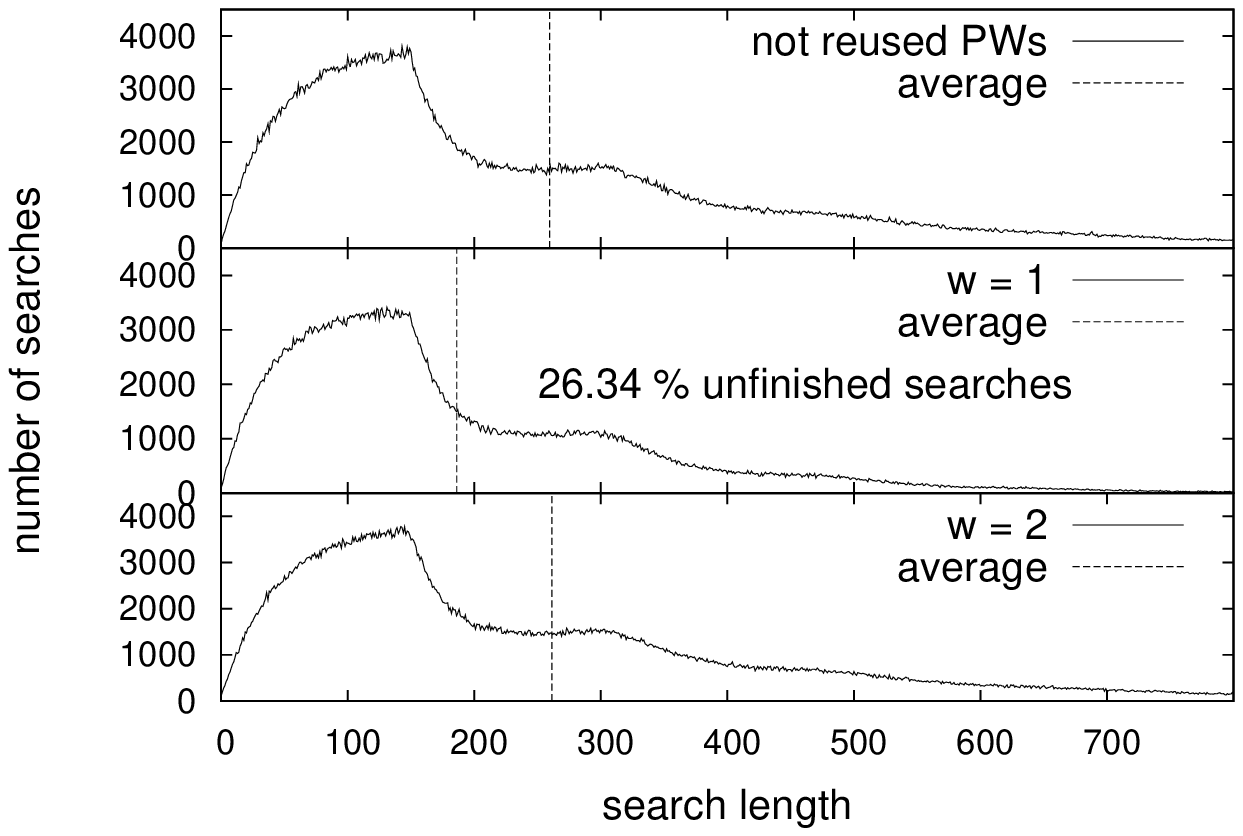}
  \label{fig:ldistr_reg10_s150_p0.01_PWinf_1_2}
 }
\qquad
 \subfloat[$p=0.1$.]{
  \includegraphics[width=\ancho cm]{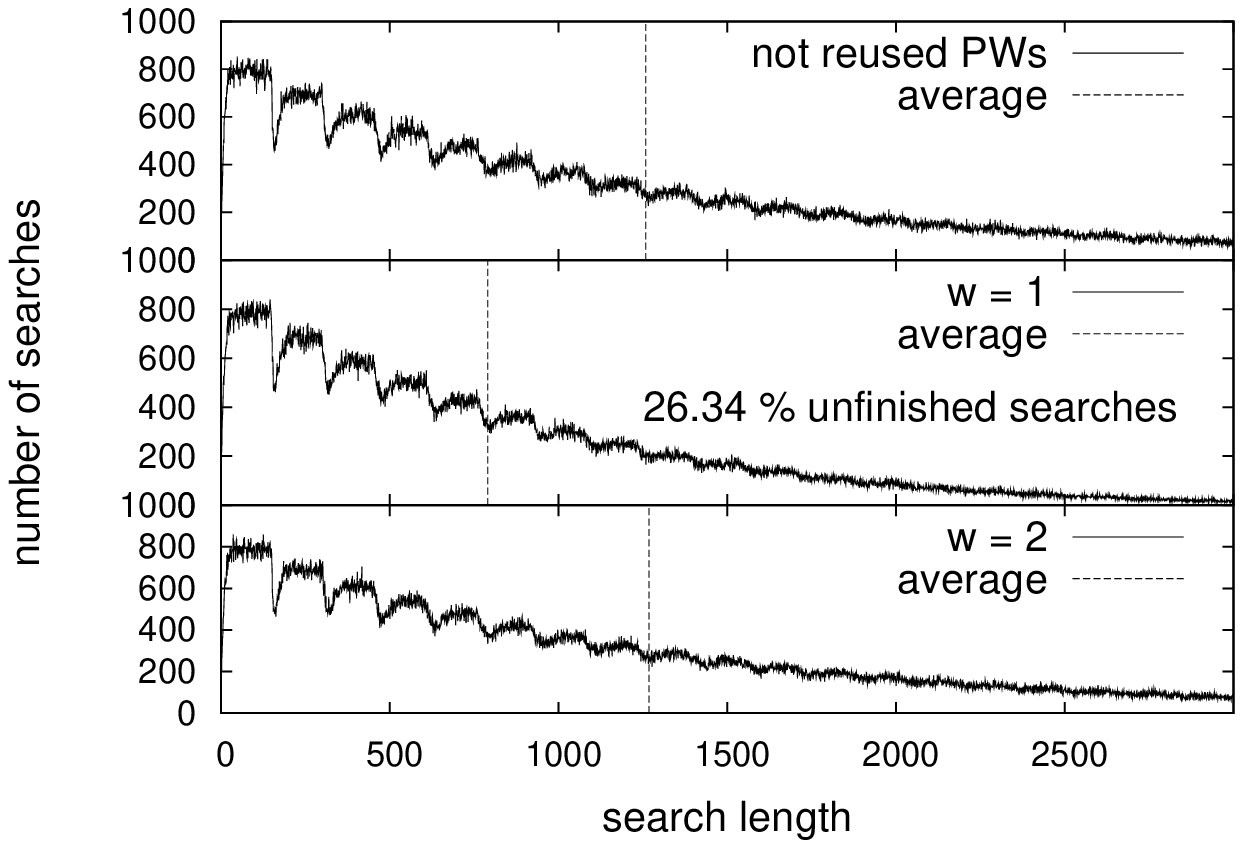}
  \label{fig:ldistr_reg10_s150_p0.1_PWinf_1_2}
 }
 \caption{Search length distributions for PWs that are not reused, for $\w=1$ and for $\w=2$, in the regular network ($\p=0,0.01,0.1$). }
 \label{fig:L_distr_reg}
\end{figure*}

\begin{figure*}
 \centering
 \subfloat[$p=0$.]{
  \includegraphics[width=\ancho cm]{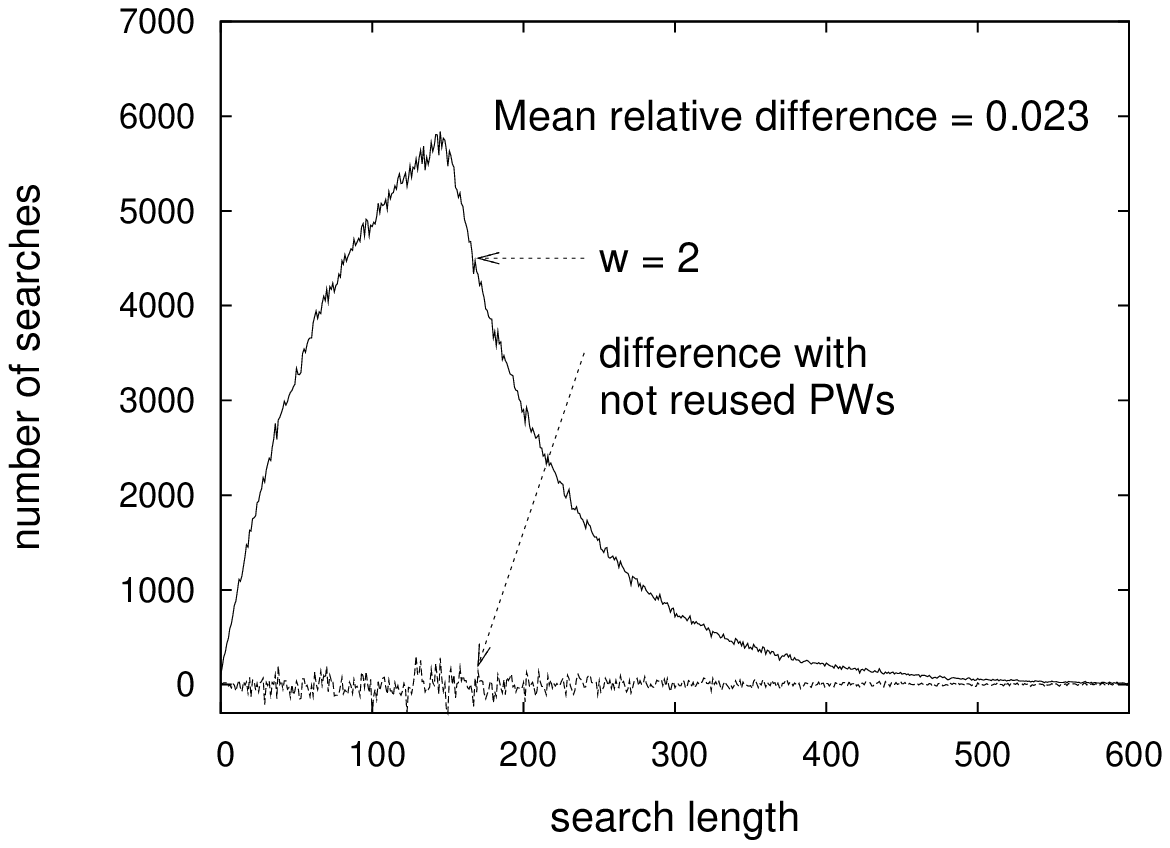}
  \label{fig:L_distr_error_reg_p0}
 }
 \subfloat[$p=0.01$.]{
  \includegraphics[width=\ancho cm]{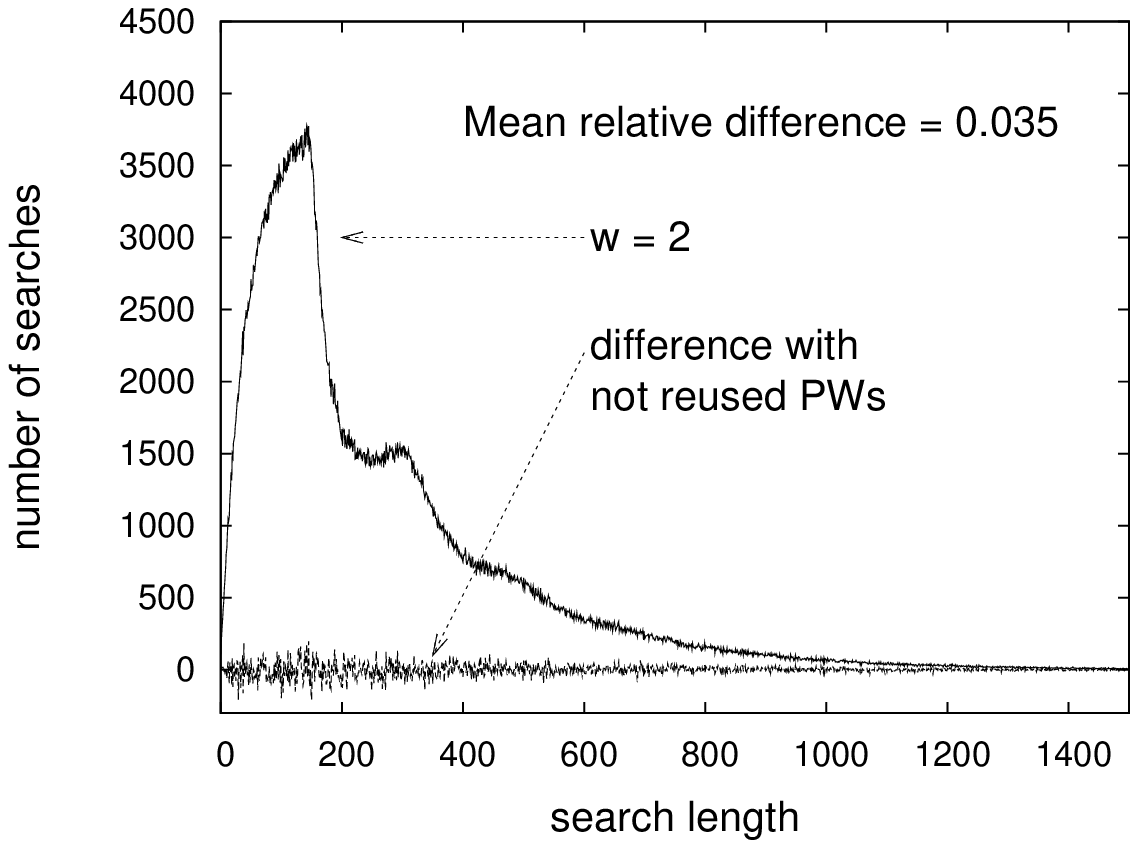}
  \label{fig:L_distr_error_reg_p0.01}
 }
\qquad
 \subfloat[$p=0.1$.]{
  \includegraphics[width=\ancho cm]{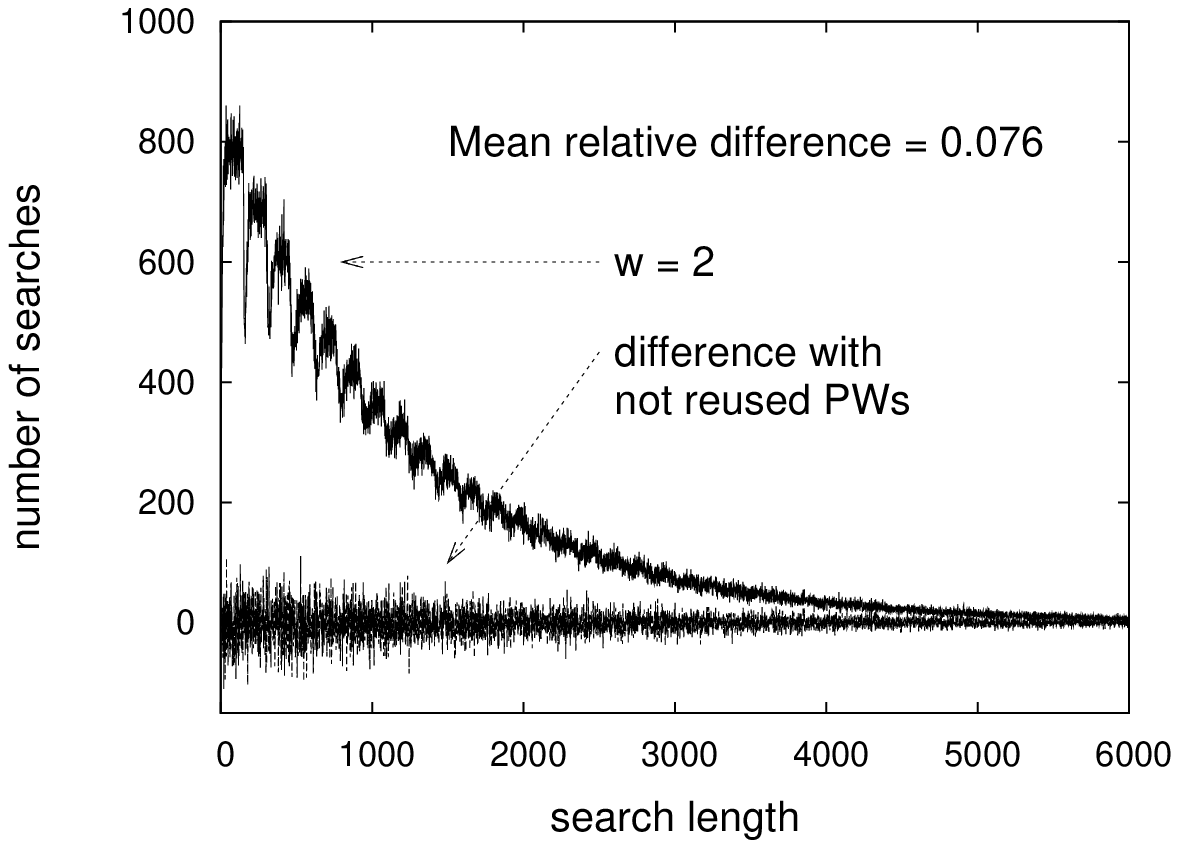}
  \label{fig:L_distr_error_reg_p0.1}
 }
 \caption{Difference between search length distributions for $\w=2$ and for not reused PWs in the regular network.}
 \label{fig:L_distr_error_reg}
\end{figure*}

We now measure the difference between the search length distributions for several values of \w\ and the base case of not reused PWs. 
In Figure~\ref{fig:L_distr_error_reg} we plot these (signed) differences for $\w=2$ and several values of $p$ in the regular network. It is observed that differences are small for low values of $p$, growing as $p$ gets bigger. But the magnitude of the differences seem to be within the order of variation of the values of the histograms for all values of $p$. As a global measure of the difference between the distributions for $\w=2$ and for PWs that are not reused we compute the \emph{mean relative difference} as
$
\frac{1}{L_{0.9}+1}\sum_{l=0}^{L_{0.9}} \frac{|h_2(l) - h_{\mathrm{af}}(l)|}{h_{\mathrm{af}}(l)},
$
where $h_w(l)$ is the number of searches with length $\ell$ when using \w\ partial walks per node, and $h_{\mathrm{af}}(l)$ corresponds to the case of not reused PWs. The tail of long searches with low frequency is removed from the calculation, since those values yield high relative differences that distort the measurement. For this, the summation includes 90\% of the searches, from length zero up to $L_{90\%}$, where $L_{90\%}$ is the 90\% percentile of search lengths. The mean relative differences for $p=0$, $p=0.01$ and $p=0.1$ are, respectively, 0.023, 0.035 and 0.076.

Therefore we conclude that, for the types of networks in our experiment, just two precomputed partial walks per node are enough to obtain searches whose lengths are statistically similar to those that would be obtained with PWs that are not reused.

\end{document}